\documentclass[journal]{IEEEtran}
\usepackage{amsmath,amsfonts, amssymb, mathrsfs}
\usepackage[pdftex]{graphicx}
\usepackage[dvips]{color}
\usepackage[center]{caption}
\usepackage{subfigure}
\usepackage{multirow,bigdelim,longtable}
\usepackage{amsthm}
\usepackage{hhline}

\usepackage{flushend}

\def\EE{\mathbb E}
\def\bH{\mathbf{H}}

\def\C{\mathcal{C}}
\def\N{\mathcal{N}} 

\def\V{\mathcal{V}}

\newcommand{\ol}[1]{\overline{{#1}}}
\DeclareMathOperator*{\argmin}{arg\,min}
\DeclareMathOperator*{\argmax}{arg\,max}
\newcommand{\squeezeup}{\vspace{0mm}}

\def\EE{\mathbb E}

\def\NNC{\text{NNC}}
\def\C{\mathcal{C}}
\def\A{\mathcal{A}}
\def\N{\mathcal{N}} 

\def\V{\mathcal{V}}
\def\L{\Omega}

\newtheorem{lemma}{Lemma}
\newtheorem{theorem}{Theorem}
\newtheorem{claim}{Claim}

\newtheorem{corollary}{Corollary}
\newtheorem{remark}{Remark}

\IEEEoverridecommandlockouts

\title{Capacity Approximations for \\Gaussian Relay Networks}

\author{Ritesh~Kolte,~\IEEEmembership{Student Member,~IEEE,} Ayfer~\"{O}zg\"{u}r,~\IEEEmembership{Member,~IEEE} and Abbas El Gamal,~\IEEEmembership{Fellow,~IEEE}    \thanks{Manuscript received July 14, 2014; accepted July 6, 2015; date of current version August 2015.  This work was presented in part in ITW 2013 Seville Spain \cite{KO13} and IZS 2014 Zurich Switzerland \cite{KOE14}.}
    \thanks{Communicated by Tie Liu, Associate Editor for Shannon Theory.}
        \thanks{The authors are with the Department of Electrical Engineering at
    Stanford University. (Emails: rkolte@stanford.edu,
    aozgur@stanford.edu, abbas@ee.stanford.edu). The work of R.~Kolte and A.~\"{O}zg\"{u}r was partly supported a Stanford Graduate Fellowship, NSF CAREER award 1254786 and the NSF Center for Science of Information under grant agreement CCF-0939370.}}%

\begin{document}
\maketitle

\begin{abstract}

Consider a Gaussian relay network where a source node communicates to a destination node with the help of several layers of relays. Recent work has shown that compress-and-forward
based strategies can achieve the capacity
of this network within an additive gap. Here, the
relays quantize their received signals at the noise level and map
them to random Gaussian codebooks. The resultant gap to capacity 
is independent of the SNR's of the channels in the network and the topology but is
linear in the total number of nodes.

In this paper, we provide an improved lower bound on the rate achieved by compress-and-forward based strategies (noisy network coding in particular) in arbitrary Gaussian relay networks, whose gap to capacity depends on the network not only through the total number of nodes but also through the degrees of freedom of the min cut of the network. We illustrate that for many networks, this refined lower bound can lead to a better approximation of the capacity. In particular, we demonstrate that it leads to a logarithmic rather than linear capacity gap in the total number of nodes for certain classes of layered networks. 
The improvement comes from quantizing the received signals of the relays at a resolution decreasing with the total number of nodes in the network. This suggests that the rule-of-thumb in literature of quantizing the received signals at the noise level can be highly suboptimal.
\end{abstract}

\begin{keywords}
  Relay Networks, Gap to Capacity, Noisy Network Coding, Network Topology, Quantization
\end{keywords}

\section{Introduction}

Consider a source node communicating to a destination node via a sequence of relays connected by point-to-point AWGN channels, as depicted in Figure~\ref{fig:line}. The capacity of this line network is achieved by simple decode-and-forward and is equal to the minimum of the capacities of the successive point-to-point links. The decoding at each stage removes the noise corrupting the information signal and therefore the end-to-end rate achieved is independent of the number of times the message is retransmitted. 

Unfortunately, the optimality of decode-and-forward is limited to this line topology, and in physically degraded networks in general. In more general networks with multiple relays at each layer, it is well-understood that the rate achieved by decode-and-forward can be arbitrarily smaller than capacity. Characterizing the capacity of more general networks  has been of interest for a long time \cite{CE79} (also see \cite{Kramer07} and references therein). Recently, significant progress has been made (\cite{ADT11,LKEC11,OD13,RV14,KH11}) which shows that compress-and-forward based strategies can be a better fit for general relay networks. Here, relays quantize/compress their observations without decoding and forward the compressions to the destination by mapping them to a new codebook. In particular, it has been shown that  compress-and-forward based relaying strategies (such as quantize-map-and-forward in \cite{ADT11} and noisy network coding in \cite{LKEC11}) can achieve rates that are within a bounded gap to the capacity of any relay network with multi-source multicast traffic. The gap is independent of the coefficients and SNR's of the constituent channels and the topology of the network. However, it depends linearly on the total number of nodes which limits the applicability of these results to small networks with a few relays. A recent result that we would like to point out here is \cite{LKK14} in which an extension of partial-decode-forward, called distributed decode-forward, has been shown to achieve a similar result. The gap to capacity for this scheme is also shown to be linear in the number of nodes, with a lower constant compared to noisy network coding.

Since the gap to capacity of compress-forward based strategies is linear in the number of nodes, for the line network in Figure~\ref{fig:line}, they yield an achievable rate whose gap to capacity is linear in the depth of the network $D$. One natural way to explain this gap is the noise accumulation. As the information signal proceeds deeper into the network, it is corrupted by more and more noise. Therefore, any strategy that does not remove the noise corrupting the signal at each stage by decoding the source message will naturally suffer a rate loss that increases with the number of stages. However, it is not clear why this rate loss should be \emph{linear} in the depth of the network as the current results in the literature suggest \cite{ADT11,LKEC11,OD13}. The total variance of the accumulated noise over the $D$ stages of the network is $D$ times the variance of the noise at each stage (assuming identical noise variances over the $D$ stages). A factor of $D$ increase in the noise variance in a point-to-point Gaussian channel would lead to at most a $\log D$ decrease in capacity, and therefore it is natural to ask if we can reduce the performance loss of compress-and-forward strategies from linear to logarithmic in $D$, first in the context of this example and then in more general networks.

\looseness=-10
\begin{figure}[t]
\centering
\includegraphics[scale=1]{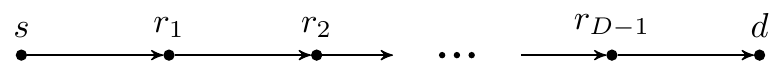}
\label{subfig:line}
\caption{Line Network 
}\label{fig:line}
\end{figure}

The first contribution of this paper is to show that a judicious choice of the quantization (or compression) resolutions at the relays can significantly improve the performance of compress-and-forward based strategies (noisy network coding in particular). For example in the line network in Figure~\ref{fig:line}, if the relay nodes quantize their observed signals at a resolution decreasing linearly in $D$, the rate loss due to compress-and-forward is only logarithmic in $D$. (See Section \ref{sec:line}.) This is counterintuitive as coarser quantization introduces more noise to the communication and our result suggests that the more relaying stages we have, the more coarsely we should quantize. The rule-of-thumb used in the current literature \cite{ADT11,LKEC11,OD13} is to quantize the received signals at the noise level (independent of the number of relays) which we show to be highly suboptimal. The improvement due to coarser quantization is because in compress-and-forward, there is a rate penalty for communicating the quantized signals to the destination and this rate penalty can be significantly larger than the rate penalty associated with coarser quantization. A detailed discussion on this is presented in Section~\ref{sec:nnc}. The fact that optimizing the quantization resolutions can lead to better rates for compress-and-forward was also observed in \cite{CO14}, \cite{SWF12} in the context of the Gaussian diamond network. 

An immediate question is whether this observation can lead to better capacity approximations for more general Gaussian networks beyond the line network. To address this question, we suggest a new approximation philosophy for the capacity of Gaussian networks. The current approach is to approximate the capacity within a gap that depends only on the number of nodes. However, two networks with the same number of nodes can have very different topologies which can potentially lead to significantly different performance for compress-and-forward. While it is desirable to have capacity approximations which are independent of the instantaneous channel realizations and SNR's in the network, since these parameters have a wide dynamical range and typically change over a short time scale in wireless networks, topological properties of a network typically change over a much longer time scale. Developing capacity approximations which reveal the dependence of the gap not only on the number of nodes but other structural properties of the network can allow for a better understanding of the performance gap of compress-and-forward strategies as well as yield tighter capacity approximations for many Gaussian networks.   

The main result of this paper is a new capacity approximation for Gaussian networks where the gap to capacity depends not only on the number of nodes but also on the number of degrees of freedom (DOF) of the mincut of the network. While the DOF of the mincut of the network can be carefully evaluated for a given network with specific channel realizations (in which case our result will yield the tightest approximation for this network), in many cases this quantity can be easily bounded based only on the topological properties of the network. For example, for the line network in Figure~\ref{fig:line} the DOF of the mincut is trivially bounded by $1$, while for a diamond network \cite{CO14} it can be trivially bounded by $2$. For such networks, our result yields a logarithmic rather than linear gap in the number of nodes. As before, the improvement is based on a judicious choice of the quantization resolutions at the relays with noisy network coding.

Finally, we look at specific settings and demonstrate that our general result can yield better capacity approximations for these settings than those available in the literature. The first setup we consider is the multi-layer fast-fading Gaussian relay network in Figure~\ref{fig:layered}. Here a source node equipped with $K$ antennas communicates to a destination node equipped with $K$ antennas over $D$ layers, each layer containing $K$ single-antenna relays. Each relay observes a noisy linear combination of the signals transmitted by the relays in the previous layer. All channels are subject to i.i.d. Rayleigh fast-fading. Current results on compress-and-forward \cite{ADT11,LKEC11,OD13} yield a rate which is within $1.3\,KD$ gap to the capacity of this network, where $KD$ is the total number of nodes. Instead, we show that if relays quantize their received signals at a resolution that decreases as the number of layers increases, compress-and-forward can achieve a rate which is within an additive gap of $K\log D+K$ of the network capacity. So for a fixed $K$, as the number of layers $D$ increases, this gap only grows logarithmically in the depth of the network $D$. 

As a side result, we provide an analysis of the compress-and-forward based strategies in \cite{ADT11,LKEC11,OD13} in fast-fading wireless networks. Fast-fading wireless networks are considered in Theorem 8.4 of \cite{ADT11}, however the conclusion of the theorem and its proof are erroneous. Theorem 8.4 of \cite{ADT11} suggests that the ergodic fast-fading capacity of a wireless relay network is approximately given by the expected value of the cutset upper bound (where the expectation is over the fading distribution). In contrast, we show that the capacity is approximately given by the minimum of the expected cut values. The difference is in the order of the expectation over the fading distribution and the minimization over different cuts. Note that the second quantity can be arbitrarily larger than the first.

\begin{figure}[!ht]
\centering
\includegraphics[scale=1]{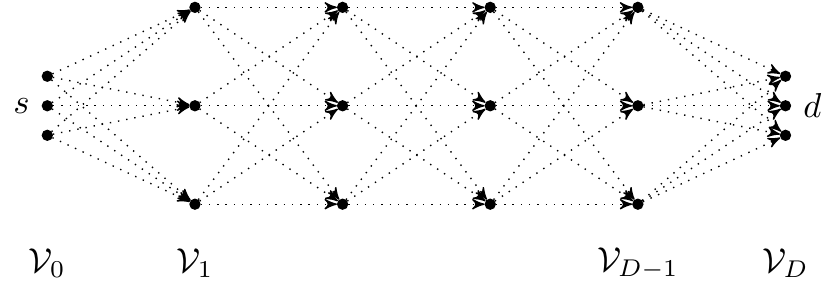}
\caption{
Multi-Layer Relay Network for $K=3$, each $H_i$ is a Rayleigh fading matrix
}\label{fig:layered}
\end{figure}

The problem of developing better capacity approximations for this  setup has also been considered in \cite{NNW11}, where a computation alignment strategy is proposed to remove the accumulating noise with the depth of the network. This yields a gap $7K^3 +5K\log K$. Computation alignment is based on the idea of combining compute-forward \cite{NG11} with ergodic alignment proposed in \cite{NGJV09}. While the gap to capacity obtained by computation alignment is independent of $D$, this strategy is significantly more complex than compress-forward and has a number of problems from a practical perspective. In particular, ergodic alignment over the fading process leads to large delays in communication and requires each relay  to know the instantaneous realizations of all the channels in the network. Moreover, its performance critically depends on the symmetry of the fading statistics. The compress-forward strategy with improved quantization we propose in this paper requires only the destination to know the instantaneous channel realizations in the network. In particular, no channel state information is required at the source and at the relays, and the fading statistics are not critical to the operation of the strategy. 

To illustrate this last point, we consider another setup where the network has the same layered topology, however the channel coefficients for each link are now fixed with unit magnitudes and arbitrary phases (i.e. each channel coefficient is of the form $e^{j\theta}$ for some arbitrary $\theta\in[0,2\pi]$). Our approximation gap for this setup is $2K^2\log D+K\log K+K$ which is again logarithmic in the depth of the network rather than linear. Computation alignment is obviously not applicable in this case and the best currently available capacity approximation for this setup is $1.3KD$ which follows from capacity approximations for general Gaussian networks \cite{ADT11,LKEC11,OD13}.

The aforementioned and previous results raise the question of whether tighter gaps scaling sublinearly in the network size can be obtained in the general case (independent of network topology). In this respect, we would like to mention an interesting recent work \cite{CO15} that shows that obtaining a gap between capacity and cutset bound that is sublinear in the number of nodes for general Gaussian relay networks is possible if and only if the cutset bound is tight for \emph{all} Gaussian relay networks.

The paper is organized as follows. The next section describes the model and some background. The main results and a discussion of the results are presented in Section~\ref{sec:main_res}. We illustrate the basic idea behind the results via the simple example of a line network in Section~\ref{sec:line}. Section~\ref{sec:nnc} aims to clarify the counterintuitive observation that coarser quantization at the relays can result in a better achievable rate. The formal proofs of the main results are presented in Sections~\ref{sec:proof1},~\ref{sec:fast_layered}~and~\ref{sec:stat_layered}. 

\section{Model and Preliminaries}\label{sec:model}
In the following subsection, we describe the general model of a Gaussian relay network, which is the subject of our main result. 
\subsection{General Model}\label{subsec:general}
Consider a Gaussian relay network, as depicted in Figure~\ref{fig:relay} where a source node $s$ communicates to a destination node $d$ a message $m\in[1:2^{nR}]$ in $n$ transmissions with the help of a set of relay nodes. Let the number of transmit antennas and receive antennas at node $i$ be $M_i$ and $N_i$ respectively. We assume $N_s=0$ and $M_d=0.$ Let $\N$ denote the set of all nodes and $M=\sum_{i\in\N}M_i$ and $N=\sum_{i\in\N}N_i$ be the total number of transmit and receive antennas respectively. The signal received by node $i$ at time $t$ is denoted as $\mathbf{Y}_i[t]\in\mathbb{C}^{N_i\times 1}$ which is given by
$$
\mathbf{Y}_i[t]=\sum_{j\neq i} \bH_{ij} \mathbf{X}_j[t]+\mathbf{Z}_i[t],
$$
where $\bH_{ij}\in\mathbb{C}^{N_i\times M_j}$ contains the (complex) channel gains from node $j$ to node $i$, and $\mathbf{X}_j[t]\in\mathbb{C}^{M_j\times 1}$ is the transmitted vector by node $j$ at time $t$. We assume that $\mathbf{Y}_{s}=0$ and $\mathbf{X}_{d}=0$. Each node is subject to an average power constraint $P$ per antenna and $\mathbf{Z}_i[t]\sim\C\N(0,\sigma^2 I)$, independent across time and across different receive antennas. The relays are constrained to be strictly causal in their operations, i.e. at any relay node $i$, $\mathbf{X}_i[t]$ can be a function only of $\{\mathbf{Y}_i[1],\mathbf{Y}_i[2],\dots ,\mathbf{Y}_i[t-1]\}.$  

A rate $R$ is said to be achievable if the probability of error of decoding the message $m\in[1:2^{nR}]$ at the destination $d$ can be made arbitrarily small by choosing a sufficiently large $n$. The supremum of all achievable rates is called the capacity $C$ of the network.

\begin{figure}[!ht]
\centering
\includegraphics[scale=1]{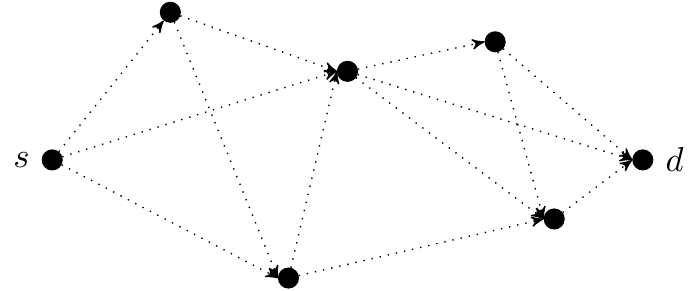}
\caption{Gaussian Relay Network}
\label{fig:relay}
\end{figure}

In sections \ref{sec:fast_layered} and \ref{sec:stat_layered}, we focus on the following two special cases of Gaussian relay networks respectively.

\subsection{Fast-fading Layered Network}\label{subsec:fast_layered}
In section \ref{sec:fast_layered}, as stated in the introduction and depicted in Figure~\ref{fig:layered}, we consider a fast-fading layered network, where each layer except the first and last contains $K$ single-antenna nodes. The nodes in the $i$th layer are collectively referred to as $\V_i$ where $0\leq i\leq D$, while a particular node $j$ in layer $i$ is referred to as the pair $(i,j)$. The layer $\V_0$ consists of the source node $s$ containing $K$ transmit antennas, while the layer $\V_D$ consists of the destination node $d$, which has $K$ receive antennas. Let $\V^{i}$ denote $\V_0\cup\V_1\cup\dots\cup \V_i$. We assume that $s$ and $d$ are equipped with multiple antennas in order to keep the problem interesting. Otherwise, the minimum cut becomes the multiple-input-single-output cut from the last layer of relays to $d$ and this trivializes the problem of approximately achieving the capacity of the network. Instead of multiple antennas at $d$, one can also assume orthogonal bit-pipes from nodes in $\V_{D-1}$ to $d$, as done in \cite{NNW11}.  

For $0\leq i\leq D-1$, the received signal at node $(i+1,j)$ in $\V_{i+1}$ (or antenna if $i=D-1$) depends only on the transmit signals of nodes in $\V_i$ and at time $t$ is given by
$$Y_{(i+1,j)}[t]=\sum_{k=1}^K h_{(i,k)\rightarrow(i+1,j)}[t]X_{(i,k)}[t]+Z_{(i+1,j)}[t],$$
The channel gain $h_{(i,k)\rightarrow(i+1,j)}$ is i.i.d. $\C\N(0,1)$ across time independent of everything else (i.e., other channel gains, noise and transmitted signals). In other words, we assume independent fast Rayleigh fading. The source nodes and the relay nodes do not know the instantaneous realizations of the channel coefficients, i.e have no transmit or receive channel state information. (The source node knows the topology of the network and the channel statistics, i.e. the end-to-end ergodic rate supported by the network.) All instantaneous channel realizations are known at the destination node and are used while decoding the transmitted message from the source node. Thus, we can effectively treat $\left\{Y_d,H\right\}$ as the received signal at the destination, where $H$ contains all the channel realizations.

\subsection{Static Layered Network }\label{subsec:stat_layered}
The topology of the static layered network that we consider in Section~\ref{sec:stat_layered} is the same as that of the fast-fading layered network, i.e. a source node with $K$ transmit antennas communicates to a destination node with $K$ receive antennas over $D-1$ layers each containing $K$ single-antenna relays. However, instead of assuming fast-fading, we now focus on the case where each channel gain $h_{(i,k)\rightarrow(i+1,j)}$ is an arbitrary complex number with unit magnitude, i.e., of the form $e^{j\theta}$ for some arbitrary $\theta\in[0,2\pi]$ (possibly different for different $(i,k)\rightarrow(i+1,j)$), where the $j$ in the superscript stands for the imaginary unit. 

\subsection{Background}\label{subsec:background}

An upper bound on the capacity $C$ of any relay network is given by the cutset bound \cite{E81}, which is as follows,
\begin{equation}\label{eq:cutset_prelim}
C \leq \ol{C} \triangleq \sup_{p(x_{\N})}\left(\min_{\L:s\in\L,d\in\L^c} \ol{C}(\L)\right),
\end{equation}
where $\L$ is a subset of $\N$, and \begin{equation}\label{eq:mutinfocut}
\ol{C}(\L) \triangleq I(X_{\L};Y_{\L^c}|X_{\L^c}),
\end{equation}
and $\L^c$ denotes $\N\setminus\L.$ The notation $X_{\L}$ is standard and refers to the set of random variables $\{X_i:i\in\L\}.$ 

In \cite{LKEC11}, the authors propose an achievability scheme based on compress-and-forward operation at the relays named ``noisy network coding'' (NNC). 
This scheme achieves any rate $R$ that is less than $R_{\NNC}$, which is given in \eqref{eq:nnc_prelim} at the top of the next page. To keep the expressions short, we are assuming that $\hat{Y}_{\L^c}$ contains $Y_d$. In other words, $\hat{Y}_d$ can be set to be equal to $Y_d.$ We refer the reader to \cite{LKEC11} for the details of this scheme. It is shown in \cite{LKEC11} that the gap between the cutset bound and the rates achieved by noisy network coding for Gaussian relay networks with multi-source multicast traffic is no more than $1.3|\N|.$

\begin{figure*}[!th]
\normalsize
\begin{equation}\label{eq:nnc_prelim}
R_{\NNC} \triangleq \sup_{\prod_{k\in\mathcal{N}}
p(x_k)p(\hat{y}_k|y_k, x_k)}\min_{\L:s\in\L,d\in\L^c} \left(I(X_\L; \hat{Y}_{\L^c}|X_{\L^c}) - I(Y_\L; \hat{Y}_\L | X_\mathcal{N}, \hat{Y}_{\L^c}) \right).
\end{equation}
\hrulefill
\end{figure*}

\section{Main Result}\label{sec:main_res}

Given a Gaussian relay network as described in Section~\ref{subsec:general} and a cut of this network $\L\subseteq\N$, for any $Q\geq 0$, we define
\begin{equation}\label{eq:C_iid}C^{i.i.d.}_{Q}(\L) \triangleq \log\det\left(I+\frac{P}{(Q+1)\sigma^2}\bH_{\L\rightarrow\L^c}\bH_{\L\rightarrow\L^c}^\dagger\right),\end{equation}
where the matrix $\bH_{\L\rightarrow\L^c}$ denotes the induced MIMO matrix from $\L$ to $\L^c$. In the case of single-antenna nodes, it is obtained by enumerating nodes in $\L$ and $\L^c$ in an arbitrary fashion and $\bH_{\L\rightarrow\L^c}$ is the $|\L^c|\times|\L|$ matrix whose $(i,j)$th entry contains the channel coefficient from node $j\in\L$ to node $i\in \L^c$. In the case of multiple antennas, it is obtained by enumerating the transmit antennas in $\L$ and receive antennas in $\L^c$ and the entries of the matrix denote the corresponding channel coefficient. In this paper, $\log$ denotes the natural logarithm. The expression in \eqref{eq:C_iid} is the mutual information across the cut $\L$, defined in \eqref{eq:mutinfocut}, when the channel input distributions at each node are i.i.d. $\C\N\left(0,PI\right)$ and the noise at each antenna is i.i.d. $\C\N(0,(Q+1)\sigma^2)$ (instead of $\C\N(0,\sigma^2)$ as originally defined in Section~\ref{subsec:general}). For a given $Q\geq 0$, let $\L_Q^*$ be the cut that minimizes $C^{i.i.d.}_Q(\L)$,
\begin{equation}\label{eq:Lq*}
\L_Q^*\triangleq\argmin_{\L:s\in\L,d\in\L^c}C^{i.i.d.}_Q(\L).
\end{equation}
Let $d_Q^*$ be the rank of the corresponding MIMO matrix $\bH_{\L_Q^*\rightarrow(\L_Q^*)^c}$. We will also refer to $d_Q^*$ as the number of degrees of freedom of the \textsc{MIMO} channel corresponding to the cut $\L_Q^*$,
expressed succinctly as \begin{equation}\label{eq:dofq}
d_Q^* = \textsc{DOF}\left(\argmin_{\L:s\in\L,d\in\L^c}C^{i.i.d.}_Q(\L)\right).
\end{equation} Note that the min cut $\L_Q^*$ and therefore $d_Q^*$ depends on $Q$. In particular, if $Q_1$ and $Q_2$ are two non-negative numbers and say $Q_1> Q_2 \geq 0$, then $d_{Q_1}^*$ can be larger than, smaller than or same as $d_{Q_2}^*$. The following theorem states our main result.

\begin{theorem}\label{thm:new_main}
The capacity $C$ of the network described in Section~\ref{subsec:general} satisfies
$$\ol{C} \geq C\geq \ol{C} - d_0^*\log\left(1+\frac{M}{d_0^*}\right) - \frac{N}{Q} - d_Q^*\log(Q+1),$$
for any non-negative $Q$, where $ \ol{C}$ is the cutset bound of the network given in \eqref{eq:cutset_prelim}.
\end{theorem}


Note that $Q$ in the theorem is a free parameter that can be optimized for a given network to minimize the gap between the achieved rate and the cutset upper bound. In the proof of the theorem, we will see that $Q$ corresponds to the variance of the quantization noise introduced at the relays in noisy network coding \cite{LKEC11}; larger $Q$ corresponds to coarser quantization. In previous works \cite{ADT11,LKEC11}, $Q$ is chosen to be constant independent of the number of nodes (or antennas) $N$ (i.e. $Q\approx 1$ and the quantization noise $Q\sigma^2$ is of the order of the Gaussian noise variance $\sigma^2$). Observe that due to the third term $N/Q$ of the gap in Theorem~\ref{thm:new_main}, this results in a gap that is at least linear in $N$. Trivially upper bounding both $d_0^*$ and $d_Q^*$  by $N$ makes the first and the third term also linear in $N$. However, in many cases, the min cut of the network can have much smaller DOF than $M$ and $N$ and in such cases allowing $Q$ to depend on $N$ can result in a much smaller gap. 

For example, in the diamond network with single-antenna at each node it is clear a priori that any cut of the network has at most two degrees of freedom, regardless of the number of relays, and therefore $d^*_Q\leq 2$ for any $Q$. It can be seen immediately from the above theorem that choosing $Q=N$ in this case results in a gap logarithmic in $N$ \cite{CO14}, which compares favorably with a gap that is linear in $N$. Similarly, for the fast-fading layered network with $K$ single-antenna nodes per layer defined in Section~\ref{subsec:fast_layered}, we show in Section~\ref{sec:fast_layered} that $d^*_Q\leq K$ for any $Q$. If there are $D$ layers in the network so that $N=M=KD$, the above expression tells us that choosing $Q$ to be proportional to $D$ gives a gap that is logarithmic in $D$ instead of linear in $D$. In Section~\ref{sec:stat_layered}, we demonstrate yet another setting in which applying Theorem~\ref{thm:new_main} and choosing $Q$ to be proportional to the number of layers allows us to obtain an improved gap. This demonstrates that the rule of thumb in the current literature to quantize received signals at the noise level ($Q\approx 1$) can be highly suboptimal.

Theorems \ref{thm:fast_layered} and \ref{thm:stat_layered} stated below provide formally the results that are mentioned in the preceding paragraph.

\begin{theorem}\label{thm:fast_layered}
The capacity $C$ of the fast-fading layered network described in Section~\ref{subsec:fast_layered} satisfies 
\begin{equation}\label{eq:thm_fast}
\ol{C} \geq C \geq \ol{C}-K\log D-K .
\end{equation}
\end{theorem}


Theorem~\ref{thm:fast_layered} follows from evaluating the required quantities in the expression in Theorem~\ref{thm:new_main} for the setup in Section~\ref{subsec:fast_layered}. However, directly applying the result of Theorem~\ref{thm:new_main} for this setup yields a gap of $2K\log D + K$. It turns out that we can further tighten the gap to $K\log D + K$ based on the observation that for this setup, the cutset bound can be evaluated explicitly and the optimal channel input distribution turns out to be independent across the antennas. The detailed proof appears in Section~\ref{subsec:fast}~and~\ref{subsec:fast_tight}.

The following corollary extends the result of Theorem~\ref{thm:fast_layered} to the setup considered in \cite{NNW11}. In this setup, instead of a single $K$-antenna source, there are $K$ single-antenna sources $\{s_1,s_2\dots, s_K\}$ interested in communicating with the destination, as depicted in Figure~\ref{fig:layered_multisrc}. We show that Theorem~\ref{thm:fast_layered} also implies a similar result for the sum-capacity $C$ of this network.
\begin{corollary}\label{cor:multisrc_fast}
The sum-capacity $C$ of the network in Figure~\ref{fig:layered_multisrc} satisfies
\begin{equation}\label{eq:thm_fast_multisrc}
\ol{C} \geq C \geq \ol{C}-K\log D-K.
\end{equation}
\end{corollary}

The proof of Corollary~\ref{cor:multisrc_fast} appears in Section~\ref{subsec:proof_cor_multisrc}.

\begin{figure}[!ht]
\centering
\includegraphics[scale=1]{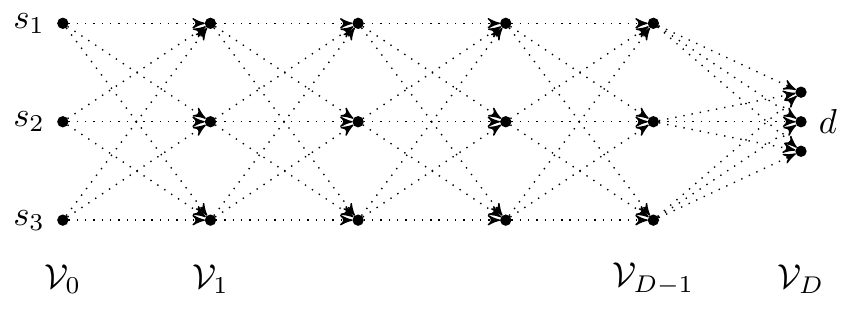}
\caption{Fast-Fading Layered Network with multiple sources}
\label{fig:layered_multisrc}
\end{figure}

The following theorem states the result for the static layered network setup, and the proof is given in Section~\ref{sec:stat_layered}.

\begin{theorem}\label{thm:stat_layered}
For $K\geq 2$ and $D\geq 2$, the capacity $C$ of the layered network described in Section~\ref{subsec:stat_layered} satisfies
\begin{equation}\label{eq:thm_stat}
\ol{C} \geq C\geq \ol{C} - 2K^2\log D  - K\log K - K .
\end{equation}
\end{theorem}

\section{Line Network}\label{sec:line}
We first illustrate the main idea of this paper in a simple setting, the line network in Figure~\ref{fig:line}. Here we assume that each link $i$ is a AWGN channel with gain $h_i$ and the channel gains $h_i$ are fixed and known. Each node has power $P$ and the noise variance is $\sigma^2$. (The conclusions below also hold under a fast-fading assumption similar to the one described in Section~\ref{sec:model}.) It is clear that a decode-forward strategy at the relays achieves the capacity of this line network, while compress-and-forward based strategies (such as quantize-map-forward in \cite{ADT11} and noisy network coding in \cite{LKEC11}) with quantization done at the noise level have a gap to capacity that is linear in the number of nodes $D$. Here, we show that if relays instead quantize at resolution $(D-1)$ times the noise level, the gap to capacity becomes logarithmic in $D$. 

Number the nodes $s$ through $d$ as $0,1,2,\dots,D$. Let's consider the rate achievable by noisy network coding  for this network, assuming all relay nodes choose their transmission codebooks independently from a Gaussian distribution, i.e. $X_i\sim\C\N(0,P)$ and independent of each other. As described in Section~\ref{subsec:background}, the rate 
$$\min_{0\leq i\leq D-1} \left(I(X_i;\hat{Y}_{i+1}|X_{i+1})-I(Y_{\V^i};\hat{Y}_{\V^i}|X_{\N},\hat{Y}_{\N\setminus\V^i})\right),$$ 
is achievable, where $\V^i=\{0,\dots,i\}$, and each relay chooses $\hat{Y}_i=Y_i+\hat{Z}_i$ where $\hat{Z}_i\sim\N(0,(D-1)\sigma^2)$ independent of everything else. Since $Y_{i+1}=h_iX_{i}+Z_{i+1}$, the channel from $X_i$ to $\hat{Y}_{i+1}$ is effectively an AWGN channel of noise power $D\sigma^2$ and gain $h_i$. Then the first term in the achievable rate expression becomes $\log\left(1+\frac{|h_i|^2P}{D\sigma^2}\right)$ which is greater than or equal to $\log\left(1+\frac{|h_i|^2P}{\sigma^2}\right)-\log(D)$.

Due to the coarse quantization, the second term in the achievable rate expression is reduced significantly as compared to quantizing at the noise level.  We have
\begin{align*}
I(Y_{\V^i};\hat{Y}_{\V^i}|X_{\N},\hat{Y}_{\N\setminus\V^i})&=I(Z_{\V^i}; \{Z+\hat{Z}\}_{\V^i})\\
&=(|\V^i|-1)\log\left(1+\frac{\sigma^2}{(D-1)\sigma^2}\right)\\
&=i\log\left(1+\frac{1}{D-1}\right)\\
& \leq \frac{i}{D-1} \\
& \leq 1,
\end{align*}
since $i\leq D-1.$ Since the capacity of the line network is given by the minimum of the capacities of each link: $\min_i \log(1+|h_i|^2P)$, we see that decreasing the resolution of quantization as the number of nodes increases results in a gap of $\log(D)+1$ to capacity. If the quantization were done at the noise level, the first term in the noisy network coding achievable rate would  suffer from only a $\log(2)$ decrease instead of $\log(D)$ with respect to capacity, however the second term would be linear in $D$, overall resulting in a gap to capacity that is linear in $D$.

At a first glance, coarser quantization resulting in better achievable rates might seem counter-intuitive. We discuss this in more depth in the following section.

\section{Gap to Capacity with Noisy Network Coding}\label{sec:nnc}
In this section, we discuss the elements of the gap between the rate achieved by noisy network coding (NNC) and the cutset bound and identify a trade-off between different elements of the gap. Our main result builds on the understanding of this trade-off.  

Consider an arbitrary discrete memoryless network with a set of nodes $\mathcal{N}$ where a source node $s$ wants to communicate to a destination node $d$ with the help of the remaining nodes acting as relays. As stated earlier in Section~\ref{subsec:background}, noisy network coding can achieve the rate given in \eqref{eq:nnc_prelim}. 
Comparing this with the cutset bound on the capacity of the network, 
\begin{equation}\label{eq:cutset}
\ol{C}=\sup_{p(x_\N)} \min_{\L:s\in\L,d\in\L^c} I(X_\L; Y_{\L^c}\,|\, X_{\L^c}),
\end{equation}
we observe the following differences. First, while the maximization in \eqref{eq:cutset} is over all possible input distributions, only independent input distributions are  admissible  in \eqref{eq:nnc_prelim}. 
\begin{figure*}[!th]
\normalsize
\begin{equation}\label{eq:NNC_refined}
\sup_{\prod_{i\in\mathcal{N}}
p(x_i)p(\hat{y}_i|y_i, x_i)}\sup_{\mathcal{M}\subseteq \mathcal{N}}\min_{\L\subseteq \mathcal{M}:s\in\L,d\in\mathcal{M}\setminus\L}\left( I(X_\L; \hat{Y}_{\L^c}|X_{\L^c}) - I(Y_\L; \hat{Y}_\L | X_\mathcal{M}, \hat{Y}_{\L^c})\right)
\end{equation}
\hrulefill
\end{figure*}
This gap corresponds to a potential beamforming gain that is allowed in the cutset bound but not exploited by NNC. Second, the first term in \eqref{eq:nnc_prelim} is similar to \eqref{eq:cutset} but with $Y_{\L^c}$ in \eqref{eq:cutset} replaced by  $\hat{Y}_{\L^c}$ in \eqref{eq:nnc_prelim}. This difference corresponds to a rate loss due to the quantization noise introduced by the relays. Third, there is the extra term $I(Y_\L; \hat{Y}_\L | X_\mathcal{N}, \hat{Y}_{\L^c})$ reducing the rate in \eqref{eq:nnc_prelim}. One way to potentially interpret this term would be as the rate penalty for communicating the quantized (compressed) observations $\hat{Y}_\L$ to the destination on top of the desired message. Note that this is the rate required to describe the observations $Y_\L$ at the distortion dictated by $\hat{Y}_\L$ to a decoder that already knows (or has decoded) $X_\mathcal{N},\hat{Y}_{\L^c}$.

However, it is not completely clear if this interpretation is precise because the non-unique decoder employed by NNC does not require the quantization indices to be explicitly decoded. The non-unique decoder of NNC searches for the unique source codeword that is jointly typical with some (not necessarily unique) set of quantization indices at the relays and the received signal at the destination. The following example in Figure~\ref{fig:1} illustrates that in certain cases the decoder can indeed recover the transmitted message even if it can not uniquely recover the quantization index of the relay. Even though we focus on the extremal case where the $r-d$ link is zero, the discussion extends to the case where this link is sufficiently weak.

\begin{figure}[!h]
	\squeezeup
	\centering
		\includegraphics{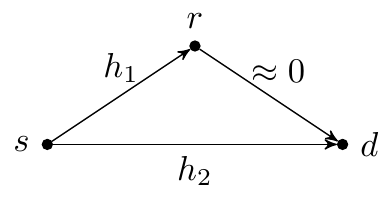}
\squeezeup
  \caption{Example\squeezeup\squeezeup}
	\label{fig:1}
\end{figure}
Consider the classical relay channel with  a very weak link from the relay to the destination. Clearly, as long as the source uses a codebook of rate less than the capacity of the direct link, no matter what the operation at the relay is, the destination can always decode the source message by performing a joint typicality test between its received signal and the source codebook (which is  subsumed by the non-unique typicality test of NNC). In particular, if the relay quantizes too finely, then there is no way for the destination to recover the relay's quantization index, even though the source message can still be recovered. 

On the other hand, this example reveals the following strange property of the expression in \eqref{eq:nnc_prelim}. While the above discussion reveals that in the setup of Fig.~\ref{fig:1}, the rate achieved by NNC is equal to the capacity of the direct link independent of the relay's operation (i.e. what $\hat{Y}_r$ is), the rate in \eqref{eq:nnc_prelim} is decreasing with increasing resolution for the quantization at the relay (due to the subtractive term $I(Y_\L; \hat{Y}_\L | X_\mathcal{N}, \hat{Y}_{\L^c})$). This suggests a more careful analysis of the rate achieved by NNC which leads to the improved rate given in \eqref{eq:NNC_refined} at the top of the next page. 
Here, only those relays that are in $\mathcal{M}\subseteq \mathcal{N}$ are considered in the non-unique typicality decoding, while the other relay transmissions are treated as noise. For example, for the relay channel in Figure~\ref{fig:1}, this would correspond to not considering the relay in the typicality decoding.

It has been shown in  \cite{WX13} that if $\mathcal{M}^*$ is the subset that maximizes \eqref{eq:NNC_refined} for a given $\prod_{i\in\mathcal{N}}
p(x_i)p(\hat{y}_i|y_i, x_i)$, then the quantization indices of the relays in $\mathcal{M}^*$ can be uniquely decoded at the destination, while the quantization indices of the relays in $\mathcal{N}\setminus \mathcal{M}^*$ cannot be decoded and in fact, it is optimal to treat the transmissions from these relays as noise. Since the transmissions from $\mathcal{N}\setminus \mathcal{M}^*$ are treated as noise, the expression \eqref{eq:NNC_refined} is increased if these relays are shut down. Hence, we can conclude that in the optimal distribution $\prod_{i\in\mathcal{N}}
p(x_i)p(\hat{y}_i|y_i, x_i)$ for NNC, some relays can be off (not utilized or equivalently always quantizing their received signals to zero) and some relays can be active, but the quantization indices of all relays (the active ones and trivially the inactive ones) can be uniquely decoded at the destination. Since the quantization indices are communicated to the destination together with the source message, there should be a rate penalty for communicating them which is precisely the term $I(Y_\L; \hat{Y}_\L | X_\mathcal{M}, \hat{Y}_{\L^c})$.

The above discussion reveals that NNC communicates not only the source message but also the quantization indices to the destination despite the non-unique typicality test performed at the decoder; and while making quantizations finer introduces less quantization noise in the communication, it leads to a larger rate penalty for communicating these quantization indices to the destination. This tradeoff is made explicit in Theorem~\ref{thm:new_main} which establishes the following achievable rate 
$$
\ol{C} - d_0^*\log\left(1+\frac{M}{d_0^*}\right) - \frac{N}{Q} - d_Q^*\log(Q+1),
$$
for any $Q\geq 0$. Here, the term $\frac{N}{Q}$ corresponds to the rate penalty associated with communicating the quantization indices and the term $d_Q^*\log(Q+1)$ corresponds to the rate penalty due to the quantization noise. Choosing a larger $Q$ increases the latter but decreases the former.

\section{Proof of Main Result}\label{sec:proof1}
In this section we prove Theorem~\ref{thm:new_main} by evaluating the rate achieved by noisy network coding in \eqref{eq:nnc_prelim} for a specific choice of the distribution $\prod_{k\in\mathcal{N}}p(x_k)p(\hat{y}_k|y_k, x_k)$ that satisfies the power constraint.
We choose the channel input vector at each node $j$ as $\mathbf{X}_j\sim \C\N\left(0,PI\right)$ and $\hat{Y}_k$ for each receive antenna in the network is chosen such that \begin{equation}\label{eq:hatY}\hat{Y}_k=Y_k + \hat{Z}_k \text{ where } \hat{Z}_k\sim\C\N(0,Q\sigma^2),\end{equation} independent of everything else, for some $Q\geq 0$. Then, the achievable rate stated in \eqref{eq:nnc_prelim} is given by
\begin{equation}\label{eq:nnc_thm1}
\min_{\L:s\in\L,d\in\L^c} \left(I(X_\L; \hat{Y}_{\L^c}|X_{\L^c}) - I(Y_\L; \hat{Y}_\L | X_\mathcal{N}, \hat{Y}_{\L^c})\right).
\end{equation}
This implies that the following rates are also achievable:
\begin{equation}\label{eq:nnc_minmax}\min_{\L:s\in\L,d\in\L^c} I(X_\L; \hat{Y}_{\L^c}|X_{\L^c}) - \max_{\L:s\in\L,d\in\L^c}I(Y_\L; \hat{Y}_\L | X_\mathcal{N}, \hat{Y}_{\L^c}).\end{equation} 

We first show that for the choice of the distribution for $X_j$'s and $\hat{Y}_k$'s in \eqref{eq:hatY}, we have $I(Y_\L; \hat{Y}_\L | X_\mathcal{N}, \hat{Y}_{\L^c}) \leq \frac{N}{Q}$ for all cuts $\Omega$ such that $s\in\L,d\in\L^c$, as follows.
\begin{IEEEeqnarray}{L}
I(Y_{\L};\hat{Y}_{\L}|X_{\N},\hat{Y}_{\L^c})\nonumber\\
 \quad =  h(\hat{Y}_{\L}|X_{\N},\hat{Y}_{\L^c})-h(\hat{Y}_{\L}|Y_{\L},X_{\N},\hat{Y}_{\L^c})\nonumber\\
\quad \stackrel{(a)}{=}  h(\hat{Y}_{\L}|X_{\N},\hat{Y}_{\L^c})-h(\hat{Y}_{\L}|Y_{\L},X_{\N})\nonumber\\
\quad \leq  h(\hat{Y}_{\L}|X_{\N})-h(\hat{Y}_{\L}|Y_{\L},X_{\N})\nonumber\\
\quad \stackrel{(b)}{=}   \left(\sum_{j\in\L}N_j\right)\log\left(Q+1\right) - \left(\sum_{j\in\L}N_j\right)\log\left(Q\right)\nonumber\\
\quad =  \left(\sum_{j\in\L}N_j\right)\log\left(1+\frac{1}{Q}\right)\nonumber\\
\quad \leq  \frac{N}{Q},\label{eq:NbyQ}
\end{IEEEeqnarray}
where both (a) and (b) follow due to our specific choice for the distribution  $\prod_{k\in\mathcal{N}}p(x_k)p(\hat{y}_k|y_k, x_k)$.
Hence,
\begin{equation}\label{eq:term2}
\max_{\L:s\in\L,d\in\L^c}I(Y_\L; \hat{Y}_\L | X_\mathcal{N}, \hat{Y}_{\L^c}) \leq \frac{N}{Q}.
\end{equation}

We now lower bound the first term in \eqref{eq:nnc_minmax}. Since $X_{\L}$ is chosen to be $\C\N(0,PI)$, the quantity $I(X_\L; \hat{Y}_{\L^c}|X_{\L^c})$ is equal to $C^{i.i.d.}_Q(\L)$, where $C^{i.i.d.}_Q(\L)$ is defined in \eqref{eq:C_iid}. Let $\L_Q^*$ denote the cut with minimal cut value as defined in \eqref{eq:Lq*}. Then,
\begin{IEEEeqnarray}{L}
\min_{\L:s\in\L,d\in\L^c} I(X_\L; \hat{Y}_{\L^c}|X_{\L^c})\nonumber\\
\quad =  \min_{\L:s\in\L,d\in\L^c} C^{i.i.d.}_Q(\L)\nonumber\\
\quad =  C^{i.i.d.}_Q(\L_Q^*)\nonumber\\
\quad \stackrel{(a)}{\geq}  C^{i.i.d.}_0(\L_Q^*) - d_Q^*\log(Q+1)\label{eq:KlogD}\\
\quad \stackrel{(b)}{\geq}   C^{i.i.d.}_0(\L_0^*) - d_Q^*\log(Q+1)\nonumber\\
\quad \stackrel{(c)}{\geq}  \sup_{p(x_\N)} I(X_{\L_0^*}; Y_{(\L_0^*)^c}\,|\, X_{(\L_0^*)^c})  \nonumber\\
\quad\quad -\> d_0^*\log\left(1+\frac{\sum_{i\in\L_0^*}M_i}{d_0^*}\right) - d_Q^*\log(Q+1)\label{eq:gapMIMO}\\
\quad \geq  \sup_{p(x_\N)} I(X_{\L_0^*}; Y_{(\L_0^*)^c}\,|\, X_{(\L_0^*)^c}) - d_0^*\log\left(1+\frac{M}{d_0^*}\right)\nonumber\\
\quad\quad -\> d_Q^*\log(Q+1)\nonumber\\
\quad =  \sup_{p(x_\N)}\min_{\L:s\in\L,d\in\L^c} I(X_\L; Y_{\L^c}\,|\, X_{\L^c}) -d_0^*\log\left(1+\frac{M}{d_0^*}\right) \nonumber\\
\quad\quad -\> d_Q^*\log(Q+1)\nonumber\\
\quad =  \ol{C}-d_0^*\log\left(1+\frac{M}{d_0^*}\right) - d_Q^*\log(Q+1)\label{eq:term1},
\end{IEEEeqnarray} where $(a)$ is justified by the following:
\begin{IEEEeqnarray}{L} 
C^{i.i.d.}_Q(\L_Q^*) \nonumber\\
\quad =  \log\det\left(I+\frac{P}{(Q+1)\sigma^2}\bH_{\L_Q^*\rightarrow(\L_Q^*)^c}\bH_{\L_Q^*\rightarrow(\L_Q^*)^c}^\dagger\right)\nonumber\\
\quad \geq  \log\det\left(I+\frac{P}{\sigma^2}\bH_{\L_Q^*\rightarrow(\L_Q^*)^c}\bH_{\L_Q^*\rightarrow(\L_Q^*)^c}^\dagger\right)\nonumber\\
\quad\quad -\> d_Q^*\log(Q+1)\nonumber\\
\quad =  C^{i.i.d.}_0(\L_Q^*) - d_Q^*\log(Q+1),\label{eq:dlogQ}
\end{IEEEeqnarray}
 $(b)$ follows by the definition of $\L_0^*$ and $(c)$ follows from \cite[Lemma 6.6]{ADT11} equation (144), which considers a MIMO channel with per-antenna power constraint and bounds the gap between its capacity and the largest achievable rate with no spatial coding, i.e. the rate achieved by using independent inputs at the antennas.

The proof of Theorem~\ref{thm:new_main} follows from \eqref{eq:term2} and \eqref{eq:term1}.\hfill\IEEEQED\\

We next state an observation which will be useful in Section~\ref{sec:stat_layered} when we analyze the static layered network.
\begin{remark}\label{remark:classA}
If there exists a set of cuts $\A$ such that $$\min_{\L:s\in\L,d\in\L^c}C^{i.i.d.}_Q(\L) \geq \min_{\substack{\L\in\A:\\ s\in\L,d\in\L^c}}C^{i.i.d.}_Q(\L) - \kappa$$ for all $Q$, where $\kappa$ is a constant, then the gap between the upper and the lower bound in Theorem~\ref{thm:new_main} can be potentially improved to 
\begin{equation}\label{eq:gapA}\tilde{d}_0^*\log\left(1+\frac{M}{\tilde{d}_0^*}\right) + \frac{N}{Q} + \tilde{d}_Q^*\log(Q+1) + \kappa,\end{equation}
where 
\begin{equation}\label{eq:dofA}\tilde{d}_Q^* \triangleq \textsc{DOF}\left(\argmin_{\substack{\L\in\A:\\ s\in\L,d\in\L^c}}C^{i.i.d.}_Q(\L)\right).\end{equation}
\end{remark}
This can be seen by modifying the proof of the lower bound \eqref{eq:term1} slightly as:
\begin{IEEEeqnarray}{L}
\min_{\L:s\in\L,d\in\L^c} I(X_\L; \hat{Y}_{\L^c}|X_{\L^c})\nonumber\\
\quad =  \min_{\L:s\in\L,d\in\L^c} C^{i.i.d.}_Q(\L) \nonumber\\
\quad\geq  \min_{\substack{\L\in\A:\\ s\in\L,d\in\L^c}}C^{i.i.d.}_Q(\L) - \kappa\nonumber\\
 \quad\geq  \min_{\substack{\L\in\A:\\ s\in\L,d\in\L^c}}C^{i.i.d.}_0(\L) - \tilde{d}_Q^*\log(Q+1) - \kappa\nonumber\\
\quad\geq  \ol{C}-\tilde{d}_0^*\log\left(1+\frac{M}{\tilde{d}_0^*}\right) - \tilde{d}_Q^*\log(Q+1) - \kappa,\nonumber
\end{IEEEeqnarray}
where each step follows by the same arguments in \eqref{eq:term1}.

\section{Fast-fading Layered Network}\label{sec:fast_layered}
In this section, we concentrate on the fast-fading layered network defined in Section~\ref{subsec:fast_layered} and obtain an approximation for the capacity of this network.

\subsection{Applying Theorem~\ref{thm:new_main} to the fast-fading layered network}\label{subsec:fast}
For the fast-fading setup, we assume that the destination knows all the instantaneous channel realizations in the network while the source and the relay nodes only know the statistics of the channel coefficients. We first note that under this assumption, the cutset bound and the noisy network coding rate can be expressed as follows.

\begin{itemize}
\item[-] Cutset Bound:\\
Noting that under the above assumption the effective received signal at the destination can be considered to be $(Y_d,H)$, where $H$ contains all the channel realizations in the network, the cutset bound in \eqref{eq:cutset_prelim} can be written as
\begin{equation}\label{eq:cutset_fast_fading}\ol{C} = \sup_{p(x_{\N})}\left(\min_{\L:s\in\L,d\in\L^c} \ol{C}(\L)\right),\end{equation}
where
\begin{IEEEeqnarray}{rCl}
\ol{C}(\L) & \triangleq & I(X_{\L};Y_{\L^c},H|X_{\L^c})\nonumber\\
& = & I(X_{\L};Y_{\L^c}|X_{\L^c},H)\nonumber
\end{IEEEeqnarray} since $X_{\N}$ is independent of $H$.

\item[-] Noisy Network Coding:\\
The rate achieved by noisy network coding is given by \eqref{eq:nnc_fastfading} given at the top of the next page, where we have again used the fact that $X_{\N}$ is independent of $H$.
\end{itemize}
\begin{figure*}[!th]
\normalsize
\begin{equation}\label{eq:nnc_fastfading}
R_{\NNC} = \sup_{\prod_{k\in\mathcal{N}}
p(x_k)p(\hat{y}_k|y_k, x_k)}\min_{\L:s\in\L,d\in\L^c} \left(I(X_\L; \hat{Y}_{\L^c}|X_{\L^c},H) - I(Y_\L; \hat{Y}_\L | X_\mathcal{N}, \hat{Y}_{\L^c},H) \right),
\end{equation}
\hrulefill
\end{figure*}

We now proceed to the proof of Theorem~\ref{thm:fast_layered}. We first note that by following similar steps as in the proof of Theorem~\ref{thm:new_main}, we can get the following result:
\begin{equation}\label{eq:thm1_thm2} \ol{C} \geq C \geq \ol{C}  - d_0^*\log\left(1+\frac{M}{d_0^*}\right) - \frac{N}{Q} - d_Q^*\log(Q+1),\end{equation} where 
$d_Q^*$ is now analogously defined as the expected degrees of freedom of the fast-fading MIMO channel corresponding to the cut $\L_Q^*$ that minimizes $\EE[C_Q^{i.i.d.}(\L)]$, which we express as
$$d_Q^*\triangleq \textsc{DOF}\left(\argmin_{\L:s\in\L,d\in\L^c}\EE\left[ C^{i.i.d.}_Q(\L)\right]\right),$$
and the expectation is with respect to the randomness in the channels. Note that when we proved Theorem~\ref{thm:new_main}, we defined $C^{i.i.d.}_Q(\L)$ to be the first mutual information term in the achievable rate for noisy network coding in \eqref{eq:nnc_thm1} when the input distributions $\mathbf{X}_j$ are i.i.d. $\C\N\left(0,PI\right)$ and $\hat{Y}_k$'s are chosen according to \eqref{eq:hatY}. In the current fast-fading case the first mutual information term in the achievable rate for noisy network coding in \eqref{eq:nnc_fastfading} is equal to $\EE[C_Q^{i.i.d.}(\L)]$ under the same distribution for the  $\mathbf{X}_j$'s and $\hat{Y}_k$'s. Therefore, the proof of Theorem~\ref{thm:new_main} can be applied verbatim in the current case by only modifying the definition of $d_Q^*$ accordingly.

Now, by choosing $Q$ to be equal to $Q'=D-1$, we get that 
\begin{IEEEeqnarray*}{rCl}
C & \geq & \ol{C} - d_0^*\log\left(1+\frac{M}{d_0^*}\right) - \frac{N}{Q'} - d_{Q'}^*\log(Q'+1)\\
& = & \ol{C} - d_0^*\log\left(1+\frac{K(D-1)}{d_0^*}\right) - \frac{K(D-1)}{Q'}\\
&& \quad -\> d_{Q'}^*\log(Q'+1)\\
& \stackrel{(a)}{=} & \ol{C} - K\log\left(1+\frac{K(D-1)}{K}\right) - \frac{K(D-1)}{Q'}\\
&&\quad -\> K\log(Q'+1)\\
& \stackrel{(b)}{\geq} & \ol{C} - K\log D - K - K\log D,\\
& = & \ol{C} - 2K\log D - K,
\end{IEEEeqnarray*}
where
\begin{itemize}
\item[-] $(a)$ follows from Lemma~\ref{lem:dq_fastfading}, provided below, which states that $d_Q^* = K$ for any $Q\geq 0$; and
\item[-] $(b)$ follows since $Q'=D-1.$
\end{itemize}
Thus, we have characterized the capacity of the fast-fading layered network within a gap of $2K\log D + K$. The next subsection describes how this result can be tightened to obtain a gap equal to $K\log D + K,$ which will conclude the proof of Theorem~\ref{thm:fast_layered}.

\begin{lemma}\label{lem:dq_fastfading}
For the fast-fading layered network,  we have for any $Q\geq 0$,
$$\min_{\L:s\in\L,d\in\L^c} \EE\left[C^{i.i.d.}_Q(\L)\right] = \EE\left[C^{i.i.d.}_Q(\V^0)\right],$$
which implies
$$d^*_Q = K.$$
\end{lemma}
\begin{proof}
See Appendix~\ref{app:d_Q}.
\end{proof}

\subsection{Tightening the approximation}\label{subsec:fast_tight}

The main idea in tightening the approximation is that for the fast-fading layered network, we can get rid of the term $d_0^*\log\left(1+\frac{M}{d_0^*}\right)$ in the gap given by Theorem~\ref{thm:new_main}. 

Recall from the proof of Theorem~\ref{thm:new_main} that this term appears because we need to bound the difference between the capacity of a MIMO channel with per-antenna power constraint and the rate achievable by using independent inputs at each antenna. However, for an i.i.d. Rayleigh fast-fading MIMO channel, it is the case that independent inputs at each node are optimal and so the largest rate achievable by using independent inputs at each antenna is equal to the capacity \cite{T99}.

Then, the proof for obtaining equation \eqref{eq:thm1_thm2} which is based on the proof of  Theorem~\ref{thm:new_main} can be repeated verbatim except for one change: in \eqref{eq:gapMIMO}, the term $d_0^*\log\left(1+\frac{\sum_{i\in\L_0^*}M_i}{d_0^*}\right)$ can be removed. This is valid since $\L_0^*=\V^0$ as shown by  Lemma~\ref{lem:dq_fastfading}, which induces an i.i.d. Rayleigh fast-fading $K\times K$ MIMO channel. 
This improves the lower bound obtained in the previous subsection from $\ol{C} - 2K\log D -K$ to $\ol{C} - K\log D -K.$ For clarity, we present the arguments in full formality below. 

We first define, for any $Q\geq 0$,
\begin{equation}\label{eq:f_Q}
f_Q(x,y)\triangleq\EE\left[\log\det\left(I + \frac{P}{(Q+1)\sigma^2}\bH_{x,y}\bH_{x,y}^\dagger\right)\right],\end{equation}
where $\bH_{x,y}$ is a $x\times y$ matrix containing i.i.d. $\C\N(0,1)$ entries. Note that using this notation, we have that $\EE\left[C^{i.i.d.}_Q(\V^0)\right]$ is equal to $f_Q(K,K)$.

Using this notation, the statement of Lemma~\ref{lem:dq_fastfading} is
\begin{equation}\label{eq:lem1_restate}\min_{\L:s\in\L,d\in\L^c} \EE\left[C^{i.i.d.}_Q(\L)\right] = \EE\left[C^{i.i.d.}_Q(\V^0)\right]=f_Q(K,K).\end{equation}

Before proceeding to the proof of the lower bound, we give the following lemma, which states that the cutset bound defined in \eqref{eq:cutset_fast_fading}, which involves a maximization over all possible input distributions, is equal to 
$\min_{\L:s\in\L,d\in\L^c} \EE\left[C^{i.i.d.}_0(\L)\right]$. 
\begin{lemma}\label{lem:cutset}
For the fast-fading layered network, 
$$\ol{C}= \min_{\L:s\in\L,d\in\L^c} \EE\left[C^{i.i.d.}_0(\L)\right],$$ and hence $\ol{C}$ also equals $\EE\left[C^{i.i.d.}_0(\V^0)\right] = f_0(K,K).$
\end{lemma}
\begin{proof}
See Appendix~\ref{app:lem}.
\end{proof}

Using the above lemma, we can now complete the proof of the tighter lower bound via the following chain of inequalities. Recall that $\mathbf{X}_j$ are chosen to be i.i.d. $\C\N\left(0,PI\right)$ and $\hat{Y}_k$'s are chosen according to \eqref{eq:hatY}. As in the previous subsection, we set $Q$ to be equal to $Q'=D-1.$
\begin{IEEEeqnarray}{rCl}
C & \stackrel{(a)}{\geq} & \min_{\L:s\in\L,d\in\L^c}\left( I(X_\L; \hat{Y}_{\L^c}|X_{\L^c},H)\right.\nonumber\\
&& \quad \left. -\> I(Y_\L; \hat{Y}_\L | X_\mathcal{N}, \hat{Y}_{\L^c},H)\right)\nonumber\\
& \geq & \min_{\L:s\in\L,d\in\L^c} I(X_\L; \hat{Y}_{\L^c}|X_{\L^c},H) \nonumber\\
&&\quad -\> \max_{\L:s\in\L,d\in\L^c} I(Y_\L; \hat{Y}_\L | X_\mathcal{N}, \hat{Y}_{\L^c},H)\nonumber\\
& \stackrel{(b)}{\geq} & \min_{\L:s\in\L,d\in\L^c} I(X_\L; \hat{Y}_{\L^c}|X_{\L^c},H) - \frac{K(D-1)}{Q'}\nonumber\\
& = & \min_{\L:s\in\L,d\in\L^c} \EE\left[C^{i.i.d.}_{Q'}(\L)\right] - \frac{K(D-1)}{Q'}\nonumber\\
& \stackrel{(c)}{=} & f_{Q'}(K,K)- \frac{K(D-1)}{Q'}\nonumber\\
& \stackrel{(d)}{\geq} & f_0(K,K) - K\log(Q'+1)- \frac{K(D-1)}{Q'}\nonumber\\
& \stackrel{(e)}{=} & \ol{C}- K\log(Q'+1)- \frac{K(D-1)}{Q'}\nonumber\\
& = & \ol{C}- K\log D- K,\label{eq:proof_thm2}
\end{IEEEeqnarray}
where
\begin{itemize}
\item[-] $(a)$ gives the rate achieved by noisy network coding,
\item[-] $(b)$ follows since, similar to \eqref{eq:NbyQ}, $$\max_{\L:s\in\L,d\in\L^c} I(Y_\L; \hat{Y}_\L | X_\mathcal{N}, \hat{Y}_{\L^c},H) \leq \frac{K(D-1)}{Q'},$$
\item[-] $(c)$ follows from \eqref{eq:lem1_restate},
\item[-] $(d)$ follows, similarly to \eqref{eq:KlogD}, because
\begin{IEEEeqnarray}{L}
f_{Q'}(K,K)\nonumber\\
\quad = \EE\left[\log\det\left(I + \frac{P}{(Q'+1)\sigma^2}\bH_{K,K}\bH_{K,K}^\dagger\right)\right]\nonumber\\
\quad \geq  \EE\left[\log\det\left(I + \frac{P}{\sigma^2}\bH_{K,K}\bH_{K,K}^\dagger\right)\right]\nonumber\\
\quad\quad -\> K\log (Q'+1)\nonumber\\
\quad =  f_0(K,K) - K\log(Q'+1), \label{eq:fQf0KlogQ}
\end{IEEEeqnarray}
\item[-] $(e)$ follows from Lemma~\ref{lem:cutset}. Note the difference between this step and the corresponding step \eqref{eq:gapMIMO} in the proof of Theorem~\ref{thm:new_main}. For general networks, the term $d_0^*\log\left(1+\frac{\sum_{i\in\N}M_i}{d_0^*}\right)$ is required, while for the special case of fast-fading layered networks, we are able to get rid of it.
\end{itemize}

This concludes the proof of Theorem~\ref{thm:fast_layered}.\hfill\IEEEQED

%

\subsection{Proof of Corollary~\ref{cor:multisrc_fast}}\label{subsec:proof_cor_multisrc}

In this subsection, we prove that the result of Theorem~\ref{thm:fast_layered} can be extended to the case with multiple sources. Assume that $K$ single-antenna sources each wish to transmit a message at rate $\frac{R}{K}$, so that the sum-rate is $R$. We have, via the cutset bound, the following upper bound on the achievable sum-rate $K$:
$$R < \sup_{p(x_{\N})}\min_{\substack{\L\; :\; s_1,s_2,\dots,s_K\in\L,\\ d\in\L^c}} I(X_\L;Y_{\L^c}|X_{\L^c},H).$$

The RHS of the above expression is equal to the cutset bound on the achievable rate in the case of a single source as given in \eqref{eq:cutset_fast_fading}. 
Hence, we have that if a sum-rate $R$ is achievable, then it must satisfy
$$R < \ol{C}.$$

This proves the upper bound on the sum-capacity. In the remainder of this subsection, we focus on proving the lower bound. As before, we fix the distribution $p(x_{\N})$ to be $\prod_{k\in\N}p(x_k)$, with each term being  $\C\N(0,P)$. The distribution $p(\hat{y}_k|y_k,x_k)$ at the relays is to be of the same form as that in \eqref{eq:hatY}. From the result for multiple sources stated in \cite[Theorem 1]{LKEC11}, we get that $R$ is achievable if for all $1\leq k\leq K$, we have
\begin{IEEEeqnarray}{rCl}
k\frac{R}{K} & < & \min_{\substack{\L : |\{s_i: s_i\in\L\}| = k, \\ d\in\L^c}} \left(I(X_\L;Y_{\L^c}|X_{\L^c},H) \right.\nonumber\\
&&\quad\quad\quad\quad\quad\quad \left. -\> I(Y_{\L};\hat{Y}_{\L}|X_{\N},Y_{\L^c},H)\right).\IEEEeqnarraynumspace\label{eq:all_constraints}
\end{IEEEeqnarray}
For a given $k$, the above constraint is obtained by considering cuts $\L$ which contain $k$ source nodes and therefore it upper bounds the sum rate $kR/K$ achievable for these $k$ sources.

Note that we get a constraint on $R$ for each value of $k$, where $k\in\{1,2,\dots,K\}$. Also, note that if we consider $k=K$, we get a constraint on $R$ that is the same as \eqref{eq:nnc_fastfading}. So, if this were the only constraint on $R$, then the proof of Theorem~\ref{thm:fast_layered} in Section~\ref{subsec:fast_tight}, which shows that the right-hand side of \eqref{eq:nnc_fastfading} is larger than $\ol{C}-K\log D - K$, would conclude the proof of Corollary~\ref{cor:multisrc_fast}. Towards this goal, we prove in Appendix~\ref{app:cor} that any $k<K$ imposes a constraint on $R$ that is only looser than the constraint
\begin{IEEEeqnarray*}{rCl}
R & < & \ol{C}-K\log D - K\\
& = & f_0(K,K) - K\log D  -K.
\end{IEEEeqnarray*}
This concludes the proof of Corollary~\ref{cor:multisrc_fast}.\hfill\IEEEQED

\section{Static Layered Networks}\label{sec:stat_layered}
In this section, we prove Theorem~\ref{thm:stat_layered}. 
We first show that for any $Q\geq 0$, $\min_{\L:s\in\L,d\in\L^c}C^{i.i.d.}_Q(\L)$ can be approximated upto an additive constant by restricting the minimization to cuts in a particular class. Then, Theorem~\ref{thm:stat_layered} is proved by making use of Remark~\ref{remark:classA}.

For convenience, let $\bH_{\V_i\rightarrow\V_{i+1}}$ denote the matrix in $\mathbb{C}^{K\times K}$ containing channel gains from nodes in layer $i$ to nodes in layer $i+1$, and call the $K^2$ entries in $\bH_{\V_i\rightarrow\V_{i+1}}$ as the links in layer $i$. With this convention in mind, let $\A$ denote the set of cuts $\L$ for which the links crossing from $\L$ to $\L^c$ come from at most $K-1$ layers, e.g. see Figure~\ref{fig:cut}.

\begin{figure}[!ht]
\centering
\includegraphics[scale=1]{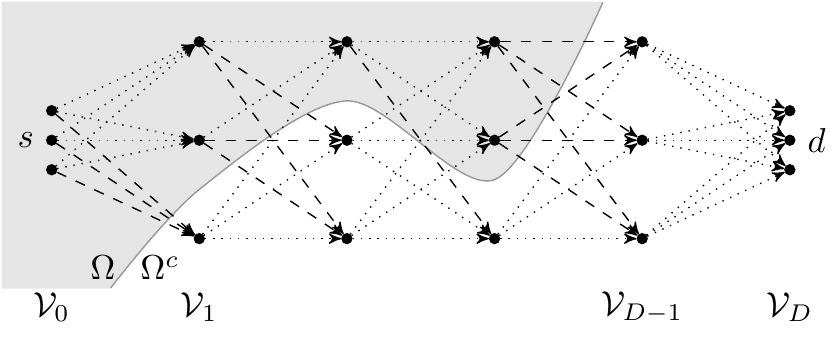}
\caption{The cut $\L$ depicted here $\notin\A$ since the crossing links come from 4 layers, and $4>K-1=2$.}\label{fig:cut}\squeezeup
\end{figure}

\begin{lemma}\label{lem:ind} For the static layered network in Section~\ref{subsec:stat_layered}, we have, for any $Q\geq 0$,
$$ \min_{\L:s\in\L,d\in\L^c} C^{i.i.d.}_Q(\L) \leq \min_{\substack{\L\in\A:\\ s\in\L,d\in\L^c}} C^{i.i.d.}_Q(\L),$$
and
$$\min_{\L:s\in\L,d\in\L^c} C^{i.i.d.}_Q(\L) \geq \min_{\substack{\L\in\A:\\ s\in\L,d\in\L^c}} C^{i.i.d.}_Q(\L)-K\log K.$$
\end{lemma}
\begin{proof}
The upper bound is immediate. The lower bound can be proved by noting that the chain of inequalities given on top of the next page, holds for any cut $\L\notin\A$, where $(a)$ follows since for any cut $\notin\A$, at least $K$ terms in the summation are non-zero and each of these terms can be lower-bounded by 
the AWGN capacity of a point-to-point channel between a single transmit and single receive antenna with unit magnitude channel coefficient; and $(b)$ follows by Lemma \ref{lem:CV0}  which is stated and proved below. This concludes the proof of the lemma.
\begin{figure*}[!th]
\normalsize
\begin{IEEEeqnarray*}{rCl}
C^{i.i.d.}_Q(\L) & = & \sum_{i=0}^{D-1} \log\det\left( I + \frac{P}{(Q+1)\sigma^2}\bH_{(\V_i\cap\L)\rightarrow (\V_{i+1}\cap\L^c)}\bH_{(\V_i\cap\L)\rightarrow (\V_{i+1}\cap\L^c)}^{\dagger} \right)\\
& \stackrel{(a)}{\geq} & K\log\left(1+\frac{P}{(Q+1)\sigma^2}\right)\\
& \stackrel{(b)}{\geq} & C^{i.i.d.}_Q(\V_0) - K\log K\\
& \geq & \min_{\substack{\L\in\A:\\ s\in\L,d\in\L^c}} C^{i.i.d.}_Q(\L) - K\log K
\end{IEEEeqnarray*}
\hrulefill
\end{figure*}
\end{proof}

\begin{lemma}\label{lem:CV0} For the static layered network in Section~\ref{subsec:stat_layered}, we have, for any $Q\geq 0,$
$$C^{i.i.d.}_Q(\V_0) \leq K\log\left(1+\frac{P}{(Q+1)\sigma^2}\right) + K\log K.$$
\end{lemma}\squeezeup
\begin{proof}
\begin{IEEEeqnarray}{RLL}
C^{i.i.d.}_Q(\V_0)& = &\log\det \left(I+\frac{P}{(Q+1)\sigma^2}\bH_{\V_0\rightarrow\V_{1}}\bH_{\V_0\rightarrow\V_{1}}^\dagger\right)\nonumber\\
& \stackrel{(a)}{\leq} & \sum_{i=1}^K\log \left(1+\frac{P}{(Q+1)\sigma^2}\mathbf{h}_i\mathbf{h}_i^\dagger\right)\nonumber\\
& \stackrel{(b)}{=} & \sum_{i=1}^K\log \left(1+\frac{P}{(Q+1)\sigma^2}K\right)\nonumber\\
& \leq & K\log\left(1+\frac{P}{(Q+1)\sigma^2}\right) + K\log K, \nonumber
\end{IEEEeqnarray}
where $\mathbf{h}_i$ denotes the $i$th row of $H_{\V_0\rightarrow\V_{1}}$ and $(a)$ follows by 
using Hadamard's inequality and $(b)$ follows from the fact that the channel gains have unit magnitude.
\end{proof}

We now use the observation made in Remark~\ref{remark:classA} to prove Theorem~\ref{thm:stat_layered}. As in the previous section, first note that $M=N=K(D-1)$. Then, we note that for any cut $\L$ in $\A$, the matrix $\bH_{\L\rightarrow\L^c}$ can have at most $K(K-1)$ columns. This is because the links crossing from $\L$ to $\L^c$ come from at most $K-1$ layers, hence there can be at most $K(K-1)$ nodes in $\L$ from which the crossing links originate. Hence, a trivial upper bound on $\tilde{d}_Q^*$ (defined in \eqref{eq:dofA}) for any $Q$ is
\begin{equation}\label{eq:bound_dA}\tilde{d}_Q^* \leq K(K-1) \leq K^2.\end{equation} 

Now, we set $Q$ to be $Q'=D-1$ and use the result in \eqref{eq:gapA} to prove Theorem~\ref{thm:stat_layered} as follows:
\begin{IEEEeqnarray*}{rCl}
C & \geq & \ol{C} - \tilde{d}_0^*\log\left(1+\frac{M}{\tilde{d}_0^*}\right) - \frac{N}{Q'} - \tilde{d}_{Q'}^*\log(Q'+1) - \kappa \\
 &\stackrel{(a)}{\geq} & \ol{C}- \tilde{d}_0^*\log\left(1+\frac{M}{\tilde{d}_0^*}\right) - \frac{N}{Q'} - \tilde{d}_{Q'}^*\log(Q'+1)\\
 &&\quad -\> K\log K\\
 &\stackrel{(b)}{\geq} &\ol{C}- K^2\log\left(1+\frac{K(D-1)}{K^2}\right) - \frac{K(D-1)}{Q'} \\
 && \quad -\> K^2\log(Q'+1) -K\log K\\
 &\stackrel{(c)}{=} &\ol{C}- K^2\log\left(1+\frac{D-1}{K}\right) - K - K^2\log D \\
 &&\quad -\> K\log K\\
 &\geq &\ol{C}- 2K^2\log D - K\log K - K,
\end{IEEEeqnarray*}
where $(a)$ follows by Lemma~\ref{lem:ind}, $(b)$ follows from \eqref{eq:bound_dA} and the fact that $x\log(1+M/x)$ is an increasing function of $x$, and $(c)$ follows since $Q'=D-1$. This concludes the proof of Theorem~\ref{thm:stat_layered}.\hfill\IEEEQED

\section{Concluding Remarks}
In this paper, we have developed improved capacity approximations for Gaussian relay networks. While existing approximations bound the capacity gap only in terms of the total number of nodes in the network, we have developed a refined approximation for the capacity of general Gaussian relay networks where the gap depends not only on the total number of nodes but other structural properties of the network (the degrees of freedom of the mincut). We have shown that this refined result allows to better approximate the capacity of many Gaussian networks, some classes of layered networks in particular.

The  improvement comes from carefully exploiting a trade-off inherent to compress-and-forward based strategies. When relays quantize/compress signals very finely, little quantization noise is introduced to the communication. When relays quantize/compress signals coarsely, there is a smaller rate penalty associated with communicating these quantization indices to the destination. We have shown that this trade-off can be very much in favor of coarse quantization, leading to the counter-intuitive principle of quantizing signals more and more coarsely with increasing number of relaying stages.

\appendices
\section{Proof of Lemma~\ref{lem:dq_fastfading}}\label{app:d_Q}
\begin{proof}
By the definition of $C^{i.i.d.}_Q(\L)$,
\begin{IEEEeqnarray*}{L}
\EE\left[ C^{i.i.d.}_Q(\L)\right] \\
\quad =\> \EE\left[\log\det\left(I+\frac{P}{(Q+1)\sigma^2}\bH_{\L\rightarrow\L^c}\bH_{\L\rightarrow\L^c}^\dagger\right)\right].
\end{IEEEeqnarray*}

We first note that for any cut $\L$ in the set $\{\V^0,\V^1,\dots,\V^{D-1}\}$, the statistics of $\bH_{\L\rightarrow\L^c}$ are identical. Hence, the value of $\EE\left[ C^{i.i.d.}_Q(\L)\right]$ is the same for all these cuts and we use $\V^0$ as a representative.

We now prove the statement: For any $Q\geq 0$,
\begin{equation}\label{eq:lem_stronger} \min_{\L:s\in\L,d\in\L^c}\EE\left[C^{i.i.d.}_Q(\L)\right] = \EE\left[C^{i.i.d.}_Q(\V^0)\right].\end{equation}
The proof of the ``$\leq$'' direction of the inequality, i.e. $$\min_{\L:s\in\L,d\in\L^c}\EE\left[C^{i.i.d.}_Q(\L)\right] \leq \EE\left[C^{i.i.d.}_Q(\V^0)\right]$$ is immediate. We focus on proving the inequality in the other direction in the remainder of this proof.

Consider a cut $\L$ that contains $M_1$ nodes from $\V_1$, $M_2$ from $\V_2$ and so on until $M_{D-1}$ from $\V_{D-1}$ (see Figure~\ref{fig:cut}). Then $\EE\left[C_Q^{i.i.d.}(\L)\right]$ is given by $$\EE\left[\log\det\left(I+\frac{P}{(Q+1)\sigma^2}\bH_{\L\rightarrow\L^c}\bH_{\L\rightarrow\L^c}^{\dagger}\right)\right],$$
where $\bH_{\L\rightarrow\L^c}$ is a block diagonal matrix containing blocks of size $M_1^c$-by-$K$, $M_2^c$-by-$M_1$, $M_3^c$-by-$M_2$, $\dots$, $M_{D-1}^c$-by-$M_{D-2}$ and finally $K$-by-$M_{D-1}$. In the preceding sentence, we have abused notation slightly by using $M_i^c$ to mean $|\V_i|-M_i=K-M_i.$

Since $\bH_{\L\rightarrow\L^c}$ has a block diagonal structure, $\EE\left[C^{i.i.d.}_Q(\L)\right]$ breaks down into a sum of terms, each being a function of the number of nodes in $\L$ that belong to two adjacent layers. Thus,
\begin{IEEEeqnarray}{L} 
\EE\left[C^{i.i.d.}_Q(\L)\right]\nonumber\\
  =  \EE\left[\log\det\left(I+\frac{P}{(Q+1)\sigma^2}\bH_{\L\rightarrow\L^c}\bH_{\L\rightarrow\L^c}^{\dagger}\right)\right]\nonumber\\
 =  f_Q(M_1^c,K)+f_Q(M_2^c,M_1)\nonumber\\
 \quad +\>\dots +f_Q(M_{D-1}^c,M_{D-2})+f_Q(K,M_{D-1}),\label{eq:sum_layers}
\end{IEEEeqnarray}
where $f_Q(x,y)$ is defined as in \eqref{eq:f_Q}:
\begin{equation*}
f_Q(x,y)\triangleq\EE\left[\log\det\left(I + \frac{P}{(Q+1)\sigma^2}\bH_{x,y}\bH_{x,y}^\dagger\right)\right],\end{equation*} and
$\bH_{x,y}$ is a $x\times y$ matrix containing i.i.d. $\C\N(0,1)$ entries. Note that using this notation, $\EE\left[C^{i.i.d.}_Q(\V^0)\right]$ is equal to $f_Q(K,K)$. So, our aim is to show that for any cut $\L$, the quantity appearing in \eqref{eq:sum_layers} is no less than $f_Q(K,K)$.

To accomplish this, we note the following properties of the function $f_Q(x,y)$:
\begin{itemize}
\item[a)] $f_Q(x,y) = f_Q(y,x)$.
\item[b)] $f_Q(z,y)\geq f_Q(x,y)$ if $z\geq x$.
\item[c)] $f_Q(x,y)+f_Q(K-x,y)\geq f_Q(K,y)$.
\end{itemize}
The first two properties are straightforward and the third property follows via a simple application of Hadamard's inequality.

Proving that the quantity in \eqref{eq:sum_layers} is no less than $f_Q(K,K)$ is just a matter of applying these properties multiple times. For concreteness, we show this for the case $D=4$ below, which can be generalized in a straightforward fashion to higher values of $D$.
\begin{IEEEeqnarray}{L}
f_Q(M_1^c,K)+f_Q(M_2^c,M_1)+f_Q(M_3^c,M_2)+f_Q(K,M_3)\nonumber\\
\geq  f_Q(M_1^c,K)+f_Q(M_2^c,M_1)\nonumber\\
\quad +\> f_Q(M_3^c,M_2)+f_Q(M_2,M_3)\nonumber\\
\geq  f_Q(M_1^c,K)+f_Q(M_2^c,M_1)+f_Q(K,M_2)\nonumber\\
\geq  f_Q(M_1^c,K)+f_Q(M_2^c,M_1)+f_Q(M_1,M_2)\nonumber\\
\geq  f_Q(M_1^c,K)+f_Q(K,M_1)\nonumber\\
\geq  f_Q(K,K)\label{eq:prop_f}\\
 = \EE\left[C^{i.i.d.}_Q(\V^0)\right],\nonumber
\end{IEEEeqnarray} 
where the first inequality follows by applying property (b) to the last term in the first line, the second inequality follows by applying (c) to the last two terms in the earlier line etc. Since this is true for any cut $\L$, we have shown that \begin{equation}\label{eq:lem1}\min_{\L:s\in\L,d\in\L^c}\EE\left[C^{i.i.d.}_Q(\L)\right]\geq \EE\left[C^{i.i.d.}_Q(\V^0)\right].\end{equation}

Thus, we have shown that \eqref{eq:lem_stronger} is true, i.e. 
\begin{equation}\label{eq:lem_stronger_restate}
\min_{\L:s\in\L,d\in\L^c}\EE\left[C^{i.i.d.}_Q(\L)\right] = \EE\left[C^{i.i.d.}_Q(\V^0)\right] = f_Q(K,K),
\end{equation}
which implies that $\V^0 \in \argmin_{\L:s\in\L,d\in\L^c} \EE\left[C^{i.i.d.}_Q(\L)\right].$ This further implies that $$d_Q^*=K,$$ since the DOF of the fast-fading MIMO channel corresponding to $\V^0$ is $K$.
\end{proof}

\section{Proof of Lemma~\ref{lem:cutset}} \label{app:lem}
Starting from \eqref{eq:cutset_fast_fading}, we have
\begin{IEEEeqnarray*}{rCl}
\ol{C} & = & \sup_{p(x_{\N})}\left(\min_{\L:s\in\L,d\in\L^c} \ol{C}(\L)\right)\\
& = & \sup_{p(x_{\N})}\left(\min_{\L:s\in\L,d\in\L^c} I(X_{\L};Y_{\L^c}|X_{\L^c},H)\right)\\
& \leq & \sup_{p(x_{\N})}\left( I(X_{\V^0};Y_{\V^0}|X_{(\V^0)^c},H)\right)\\
& \stackrel{(a)}{=} & \EE\left[\log\det\left(I+\frac{P}{\sigma^2}\bH_{\V^0\rightarrow(\V^0)^c}\bH_{\V^0\rightarrow(\V^0)^c}^\dagger\right)\right]\\
& = & \EE\left[C^{i.i.d.}_0(\V^0)\right]\\
& \stackrel{(b)}{=} & \min_{\L:s\in\L,d\in\L^c} \EE\left[C^{i.i.d.}_0(\L)\right]\\
& \leq & \sup_{p(x_{\N})}\left(\min_{\L:s\in\L,d\in\L^c} I(X_{\L};Y_{\L^c}|X_{\L^c},H)\right)\\
&= & \ol{C},
\end{IEEEeqnarray*}
where  $(a)$ follows by the fact that for a i.i.d. Rayleigh fast-fading MIMO channel, the optimal input distribution is independent across antennas \cite{T99}, and $(b)$ follows from \eqref{eq:lem_stronger} which shows that the cut that minimizes $\EE\left[C^{i.i.d.}_0(\L)\right]$ is $\V^0$.\hfill\IEEEQED

\section{}\label{app:cor}
In this appendix, we elaborate on the argument required to prove the lower bound in Corollary~\ref{cor:multisrc_fast}. 

Consider a cut $\L$ such that $|\{s_i: s_i\in\L\}| = k.$ Let $\L$ contain $M_i$ nodes from layer $\V_i$, for $1\leq i\leq D-1$. As before, we choose the quantization noise variance $Q$ to be $Q'=D-1$. This gives us a constraint on the achievable sum-rate $R$ as follows:
\begin{IEEEeqnarray*}{rCl}
R & < & \frac{K}{k}\left(I(X_\L;\hat{Y}_{\L^c}|X_{\L^c},H) - I(Y_{\L};\hat{Y}_{\L}|X_{\N},\hat{Y}_{\L^c},H)\right)\\
& = & \frac{K}{k}\left(\EE\left[C^{i.i.d.}_{Q'}(\L)\right] - I(Y_{\L};\hat{Y}_{\L}|X_{\N},\hat{Y}_{\L^c},H)\right)\\
& = & \frac{K}{k}\Big( f_{Q'}(M_1^c,k) + f_{Q'}(M_2^c,M_1) + \dots + f_{Q'}(K,M_{D-1})\\
&&\quad  -\> I(Y_{\L};\hat{Y}_{\L}|X_{\N},\hat{Y}_{\L^c},H) \Big),
\end{IEEEeqnarray*}
where we use the notation $f_Q(x,y)$ defined in \eqref{eq:f_Q}. Since we have $$ I(Y_{\L};\hat{Y}_{\L}|X_{\N},\hat{Y}_{\L^c},H) \leq \frac{\sum_{i=1}^{D-1}M_i}{Q'}= \frac{\sum_{i=1}^{D-1}M_i}{D-1},$$ which can be proved using steps similar to those used to arrive at \eqref{eq:NbyQ},  we can impose a tighter constraint on the sum-rate $R$ due to the cut $\L$, which is as follows.
\begin{IEEEeqnarray}{rCl}
R & < & \frac{K}{k}\bigg(f_{Q'}(M_1^c,k) + f_{Q'}(M_2^c,M_1)\nonumber\\
&& \quad  +\> \dots + f_{Q'}(K,M_{D-1})- \frac{\sum_{i=1}^{D-1}M_i}{D-1}\bigg) .\label{eq:constraint_k}
\end{IEEEeqnarray}

In the following, we show for any $k<K$, the above is weaker than
\begin{IEEEeqnarray}{rCl}
R & < & f_{0}(K,K)-K\log D - K, \label{eq:constraint_k_K}
\end{IEEEeqnarray} 
i.e. the right-hand side of \eqref{eq:constraint_k} for any $k<K$ is larger than $f_{0}(K,K)-K\log D - K$. 

Note that if $f_{0}(K,K)-K\log D -K\leq 0$, the achievable rate claimed by \eqref{eq:constraint_k_K} is zero so there is nothing to prove, so we assume that $f_{0}(K,K)-K\log D - K>0$.



\begin{figure*}[!th]
\normalsize
\begin{IEEEeqnarray*}{C}
\frac{1}{k~{K \choose k}} \sum_{1\leq i_1 < \dots < i_k\leq K}\log\det \left(\pi e \left(I_k+\lambda\,\bH_{l,(i_1,\dots,i_k)}^\dagger\bH_{l,(i_1,\dots , i_k)}\right)\right) \, \geq\, \frac{1}{K}\log\det \left(\pi e\left(I_K +\lambda\,\bH_{l,K}^\dagger\bH_{l,K}\right)\right)
\end{IEEEeqnarray*}
\hrulefill
\end{figure*}

\begin{itemize}
\item If the cut $\L$ has $M_{1} = M_2 = \dots  = M_{D-1} = 0$, then the expression in the constraint \eqref{eq:constraint_k}  becomes 
\begin{IEEEeqnarray*}{LCl}
 \frac{K}{k}\bigg(f_{Q'}(M_1^c,k) + f_{Q'}(M_2^c,M_1) \\
 \quad\quad\quad +\> \dots + f_{Q'}(K,M_{D-1})- \frac{\sum_{i=1}^{D-1}M_i}{D-1}\bigg)\\
\quad =  \frac{K}{k}f_{Q'}(K,k)\\
\quad \stackrel{(a)}{\geq}  f_{Q'}(K,K)\\
\quad \geq  f_{Q'}(K,K) - K\\
\quad \stackrel{(b)}{\geq} f_0(K,K) - K\log D - K,
\end{IEEEeqnarray*} 
where $(a)$ follows from Claim~\ref{claim:gen_hadamard}, provided at the end of this Appendix, and $(b)$ follows by the same argument as in \eqref{eq:fQf0KlogQ}.
\item If the cut is $\L$ such that $M_i = K$ for some $i\in\{1,2,\dots ,K\}$, then 
\begin{IEEEeqnarray*}{lCl}
\frac{K}{k}\bigg(f_{Q'}(M_1^c,k) + f_{Q'}(M_2^c,M_1) \\
\quad\quad\quad +\> \dots + f_{Q'}(K,M_{D-1})- \frac{\sum_{i=1}^{D-1}M_i}{D-1}\bigg)\\
\quad \stackrel{(a)}{\geq} \frac{K}{k}\left(f_{Q'}(K,K)- \frac{\sum_{i=1}^{D-1}M_i}{D-1}\right)\\
\quad \geq \frac{K}{k}\left(f_{Q'}(K,K)- K\right)\\
\quad \stackrel{(b)}{\geq}  f_{Q'}(K,K)- K\\
\quad \geq f_0(K,K) - K\log D -K,
\end{IEEEeqnarray*} 
where $(a)$ follows by using the properties of the function $f_Q$ as in \eqref{eq:prop_f}, and $(b)$ follows since $\frac{K}{k}\geq 1.$
\item Let $i^* = \argmax_{1\leq i\leq D-1} M_i$ so that $M_{i^*} = \max_{1\leq i \leq D-1} M_i$. From the previous two cases, we can focus our attention to $0<M_{i^*} < K.$ Also, note that $M_1<K$ implies that $M_1^c > 0.$ The RHS of the constraint due to $\L$ is
\begin{IEEEeqnarray*}{LCl}
 \frac{K}{k}\bigg(f_{Q'}(M_1^c,k) + f_{Q'}(M_2^c,M_1)\\
 \quad\quad\quad +\> \dots + f_{Q'}(K,M_{D-1})- \frac{\sum_{i=1}^{D-1}M_i}{D-1}\bigg)\\
\quad =  \frac{K}{k}f_{Q'}(M_1^c,k) + \frac{K}{k}\bigg(f_{Q'}(M_2^c,M_1) \\
\quad\quad\quad\quad +\> \dots + f_{Q'}(K,M_{D-1})- \frac{\sum_{i=1}^{D-1}M_i}{D-1}\bigg)\\
\quad \stackrel{(a)}{\geq}  f_{Q'}(M_1^c,K) + \frac{K}{k}\bigg(f_{Q'}(M_2^c,M_1) \\ \quad\quad\quad\quad +\> \dots + f_{Q'}(K,M_{D-1})- \frac{\sum_{i=1}^{D-1}M_i}{D-1}\bigg)\\
\quad \stackrel{(b)}{\geq}   f_{Q'}(M_1^c,K) + \bigg(f_{Q'}(M_2^c,M_1) \\
\quad\quad\quad\quad +\> \dots + f_{Q'}(K,M_{D-1})- \frac{\sum_{i=1}^{D-1}M_i}{D-1}\bigg)\\
\quad \stackrel{(c)}{\geq} f_{Q'}(K,K) - K\\
\quad \geq f_0(K,K) - K\log D - K,
\end{IEEEeqnarray*}
where 
\begin{itemize}
\item[-] $(a)$ follows by Claim~\ref{claim:gen_hadamard}, 
\item[-] $(b)$ follows because $\frac{K}{k}\geq 1$ and because
$$f_{Q'}(M_2^c,M_1) + \dots + f_{Q'}(K,M_{D-1})- \frac{\sum_{i=1}^{D-1}M_i}{D-1},$$is non-negative, which is proved as follows:
\begin{IEEEeqnarray*}{l}
f_{Q'}(M_2^c,M_1) + \dots + f_{Q'}(K,M_{D-1})- \frac{\sum_{i=1}^{D-1}M_i}{D-1} \\ \quad\quad\geq  f_{Q'}(K,M_{i^*}) - \frac{\sum_{i=1}^{D-1}M_i}{D-1}\\
\quad\quad \geq  f_{Q'}(K,M_{i^*}) - M_{i^*}\\
\quad\quad \geq  \frac{M_{i^*}}{K}f_{Q'}(K,K) - M_{i^*}\\
\quad\quad = \frac{M_{i^*}}{K}\left(f_{Q'}(K,K)-K\right)\\
\quad\quad \geq  \frac{M_{i^*}}{K}\left(f_{0}(K,K)-K\log D - K\right)\\
\quad\quad \geq  0,
\end{IEEEeqnarray*}
\item[-] $(c)$ follows by noting that the expression in $(b)$ is the constraint on sum-rate imposed by a cut which is $\V_0\cup \L$, which we know is lower bounded by $f_{Q'}(K,K)-K.$ 
\end{itemize}
\end{itemize}

The above analysis shows that \eqref{eq:constraint_k_K} renders all other constraints redundant.\hfill\IEEEQED

\begin{claim}\label{claim:gen_hadamard}For any $Q\geq 0,$ any $k\in\{1,2,\dots, K-1\}$ and any $l\in\{1,2,\dots, K\}$,
$$\frac{K}{k}f_Q(l,k) \geq f_Q(l,K).$$
\end{claim}
\begin{proof}
Recall that $f_Q(l,K)$ is defined to be $$\EE\left[\log\det\left(I + \frac{P}{(Q+1)\sigma^2}\bH_{l,K}^\dagger\bH_{l,K}\right)\right].$$

To be more explicit in the following, we write $I_p$ to denote an identity matrix of size $p$. Also, for brevity, we denote $\frac{P}{(Q+1)\sigma^2}$ by $\lambda$. For any fixed $\bH_{l,K}$, we have by \cite[eq. (3.15)]{H78} the inequality given at the top of this page, where $\bH_{l,(i_1,\dots,i_k)}$ is obtained by choosing the columns of $\bH_{l,K}$ indexed by $(i_1,\dots, i_k).$

Hence,
\begin{IEEEeqnarray*}{l}
\frac{1}{k~{K \choose k}}\hspace{-3.5pt} \sum_{1\leq i_1 < \dots < i_k\leq K}\log\det  \left(I_k+\lambda\,\bH_{l,(i_1,\dots,i_k)}^\dagger\bH_{l,(i_1,\dots , i_k)}\right)\\
\quad\quad\quad\quad +\> \frac{1}{k}\log\left((\pi e)^k\right)\\
\quad \geq \frac{1}{K}\log\left((\pi e)^K\right) + \frac{1}{K}\log\det \left(I_K+\lambda\,\bH_{l,K}^\dagger\bH_{l,K}\right),
\end{IEEEeqnarray*}
which means
\begin{IEEEeqnarray*}{L}
\frac{1}{k~{K \choose k}} \sum_{1\leq i_1 < \dots < i_k\leq K}\log\det  \left(I+\lambda\,\bH_{l,(i_1,\dots,i_k)}^\dagger\bH_{l,(i_1,\dots , i_k)}\right)\\
\quad\quad\quad\quad \geq  \frac{1}{K}\log\det\left(I+\lambda\,\bH_{l,K}^\dagger\bH_{l,K}\right).
\end{IEEEeqnarray*}

Now, taking expectation on both sides and observing that each term in the summation has identical statistics, the desired claim is proved.
\end{proof}

\bibliographystyle{IEEEtran}
\bibliography{IEEEfull,itw_gap}

\end{document}